\documentclass[conference,a4paper]{IEEEtran}

\def\naturals{\mathbb{N}}
\def\reals{\mathbb{R}}
\def\naturals{\mathbb{N}}
\def\define{:{=}~}
\def\BinaryField{\mathcal{F}_{2}}

\def\SourceAlphabet{\mathcal{S}}
\def\MACInputAlphabet{\mathcal{X}}
\def\MACOutputAlphabet{\mathcal{Y}}
\def\SourceAlphabet{\mathcal{S}}
\def\underlineSourceAlphabet{\underline{\SourceAlphabet}}
\def\ReconstructionAlphabet{\mathcal{S}}
\def\underlineS{\underline{S}}
\def\underlineT{\underline{T}}
\def\underlineMACInputAlphabet{\underline{\MACInputAlphabet}}

\def\underlineS{\underline{S}}
\def\underlines{\underline{s}}
\def\underlineX{\underline{X}}
\def\underlineV{\underline{V}}
\def\underlineU{\underline{U}}

\def\MessageSetM{\mathcal{M}}
\def\underlinee{\underline{e}}

\def\hatm{{\hat{m}}}
\def\hata{{\hat{a}}}
\def\hatv{{\hat{v}}}
\def\fieldq{\mathcal{F}_{q}}
\def\setU{\mathcal{U}}
\def\setX{\mathcal{X}}

\usepackage{amsmath}
\usepackage{amssymb}
\usepackage{slashbox}
\usepackage{mathrsfs}
\newcommand{\comment}[1]{}
\begin{document}

\newtheorem{remark}{\it Remark}
\newtheorem{thm}{Theorem}
\newtheorem{definition}{Definition}
\newtheorem{lemma}{Lemma}
\newtheorem{example}{\it Example}

\title{Computing sum of sources over an arbitrary multiple access channel}

\author{
  \IEEEauthorblockN{Arun Padakandla}
  \IEEEauthorblockA{ University of Michigan\\
    Ann Arbor, MI 48109, USA\\
    Email: arunpr@umich.edu} 
  \and
  \IEEEauthorblockN{S. Sandeep~Pradhan}
  \IEEEauthorblockA{    University of Michigan\\
    Ann Arbor, MI 48109, USA\\
    Email: pradhanv@eecs.umich.edu}}
    
\maketitle
\begin{abstract}
The problem of computing sum of sources over a multiple access channel (MAC) is considered. Building on the technique of linear computation coding (LCC) proposed by Nazer and Gastpar \cite{200710TIT_NazGas}, we employ the ensemble of nested coset codes to derive a new set of sufficient conditions for computing sum of sources over an \textit{arbitrary} MAC. The optimality of nested coset codes \cite{201108ISIT_PadPra} enables this technique outperform LCC even for linear MAC with a structural match. Examples of non-additive MAC for which the technique proposed herein outperforms separation and systematic based computation are also presented. Finally, this technique is enhanced by incorporating separation based strategy, leading to a new set of sufficient conditions for computing sum over a MAC.
\end{abstract}
\section{Introduction}
\label{Sec:Introduction}
Consider a scenario wherein a centralized receiver is interested in evaluating a
multi-variate function, the arguments of which are available to spatially distributed
transmitters. Traditionally, the technique of computing functions at a centralized receiver is based on
it's decoding of the arguments in it's entirety. Solutions based on this technique have
been proven optimal for particular instances of distributed source coding. Moreover, this
technique lends itself naturally for communication based on separation. Buoyed by this
partial success and ease of implementation, the de facto framework for computing at a
centralized receiver is by enabling the decoder decode the arguments of the function in
it's entirety.

The problem of computing mod-$2$ sum of distributed binary sources has proved to be an
exception. Studied in the context of a source coding problem, K\"orner and Marton
\cite{197903TIT_KorMar} propose an ingenious technique based on linear codes, that
circumvent the need to communicate sources to the decoder, and thereby
perform strictly better for a class of source distributions. In fact, as proposed in \cite{197903TIT_KorMar}, the decoder needs only sum of message indices put out by the
source encoder. This fact has been further exploited by Nazer and Gastpar
\cite{200710TIT_NazGas} in developing a channel coding technique for a \textit{linear}
MAC, henceforth referred to
as linear computation coding (LCC), that enables the decoder
reconstruct the sum of the message indices input to the channel encoder. Since the decoder
does not need to disambiguate individual message indices, this technique, when
applicable, outperforms earlier known techniques.

LCC \cite{200710TIT_NazGas} is built around employing the same linear code as a channel
code at both encoders. The message indices output by the K\"orner-Marton (KM) source code is
linearly mapped into channel codewords. Since a linear MAC first computes a sum of the
transmitted codewords, it is as if the codeword corresponding to the sum of messages was
input to the ensuing channel. The first question that comes to mind is the following. If
the MAC is not linear, would it be possible to decode sum of message indices without
having to decode the individual codewords? In other words, what would be the
generalization of LCC for an arbitrary MAC?\footnote{The technique of systematic
computation coding (SCC) \cite{200710TIT_NazGas} may not be considered as a generalization
of LCC. Indeed SCC does not reduce to LCC for a linear MAC.} If there exist such a
generalization, how
efficient would it be?

In this article, we answer the above question in the affirmative. Firstly, we recognize
that in order to decode the sum of transmitted codewords, it is most efficient to employ
channel codes that are closed under addition, of which a linear code employed in LCC is the simplest
example. Closure under addition contains the range of the sum of transmitted codewords and
thereby support a larger range for individual messages. Secondly, typical set decoding
circumvents need for the MAC to be linear. Since nested coset codes have been proven to
achieve capacity of arbitrary point-to-point channels \cite{201108ISIT_PadPra} and are
closed under addition, we
employ this ensemble for generalizing the technique of LCC. As illustrated by examples
\ref{Ex:PentaAdditiveMACWithTwoNoisyAndThreeNoiselessSymbols},\ref{Ex:PentaAdditiveMACWithTwoUselessAndThreeNoisySymbols} in
section \ref{Sec:NestedCosetCodesForComputingTheSumOfSourcesOverAMAC}, the generalization
we propose (i) outperforms separation based
technique for an arbitrary MAC and moreover (ii) outperforms LCC even for examples with a
structural match.\footnote{This is expected since linear codes achieve only symmetric
capacity and nested coset codes can achieve capacity of arbitrary point-to-point
channels.} We remark that analysis of typical set decoding of a function of transmitted
codewords with nested coset codes that contain statistically dependent codewords
contains new elements and are detailed in proof of theorem \ref{Thm:AchievabilityUsingNestedCosetCodes}.

Even in the case of a structural match, separation based schemes might outperform LCC
\cite[Example 4]{200710TIT_NazGas}. This raises the following question. What then would be
a unified scheme for computing over
an arbitrary MAC? Is there such a scheme that reduces to (i) separation when the desired
function and MAC are not matched and (ii) LCC when appropriately matched? We recognize
that KM technique is indeed suboptimal for a class of source
distributions. For such sources, it is more efficient to communicate sources as is. We
therefore take the approach of Ahlswede and Han \cite[Section VI]{198305TIT_AhlHan}, where
in a two layer source code accomplishes distributed compression. The first layer generates
message indices of those parts that are best reconstructed as is, and the second
employs a KM technique. In section \ref{Sec:GeneralTechnique}, we propose a two layer
channel code for MAC that is compatible with the above two layer source code. The first layer of
the MAC channel code communicates the message indices as is, while the second
enables the decoder decode the sum of second layer message indices, and thereby develop a
unifying strategy that subsumes separation and LCC.

We highlight the significance of our contribution. Firstly, we propose a strategy based on nested coset codes and derive a set of sufficient conditions for the problem of computing sum of sources over an \textit{arbitrary} MAC. The proposed strategy subsumes all current known strategies and performs strictly better for certain examples (section \ref{Sec:NestedCosetCodesForComputingTheSumOfSourcesOverAMAC}). Secondly, our findings highlight the utility of nested coset codes
\cite{201108ISIT_PadPra} as a generic ensemble of structured codes for communicating over
arbitrary multi-terminal communication problems. Thirdly, and perhaps more importantly,
our findings hint at a general theory of structured codes. Linear and nested linear codes
have been employed to derive communication strategies for particular symmetric
additive source and channel coding problems that outperform all classical unstructured-code based techniques. However the question remains whether these structured code based
techniques can be generalized to arbitrary multi-terminal communication problems. Our
findings indicate that strategies based on structured
codes can be employed to analyze more intelligent encoding and decoding techniques for
an arbitrary multi-terminal communication problem.

\section{Preliminaries and Problem statement}
\label{Sec:PreliminariesProblemStatement}
Following remarks on notation (\ref{SubSec:Notation}) and problem statement (\ref{SubSec:ProblemStatement}), we briefly describe LCC for a linear MAC (\ref{SubSec:LinearComputationCoding}) and set the stage for it's generalization.

\subsection{Notation}
\label{SubSec:Notation}
We employ notation that is now widely adopted in the information theory literature
supplemented by the following. We let $\fieldq$ denote a finite field of cardinality
$q$. When the finite field is clear from
context, we let $\oplus$ denote addition in the same. When ambiguous, or to enhance clarity, we specify addition in $\fieldq$ using
$\oplus_{q}$.
In this article, we repeatedly refer to pairs of objects of similar type. To reduce
clutter in notation, we use an \underline{underline} to refer to aggregates of similar
type. For example, (i) $\underlineS$ abbreviates $(S_{1},S_{2})$, (ii) if $\MACInputAlphabet_{1},\MACInputAlphabet_{2}$
are finite alphabet sets, we let $\underlineMACInputAlphabet$ either denote the Cartesian
product $\MACInputAlphabet_{1} \times \MACInputAlphabet_{2}$ or abbreviate the pair
$\MACInputAlphabet_{1},\MACInputAlphabet_{2}$ of sets. More non trivially, if $e_{j}:\SourceAlphabet^{n} \rightarrow
\MACInputAlphabet_{j}^{n}:j=1,2$ are a pair of maps, we let $\underlinee(\underlines^{n})$ abbreviate
$(e_{1}(s_{1}^{n}),e_{2}(s_{2}^{n}))$.
\subsection{Problem statement}
\label{SubSec:ProblemStatement}
Consider a pair $(S_{1},S_{2})$ of information sources each taking values
over a finite field $\SourceAlphabet$ of cardinality $q$. We assume outcome $(S_{1,t},S_{2,t})$ of the sources
at time $t \in \mathbb{N}$, is independent and
identically
distributed across time, with distribution $W_{\underlineS}$. We let
$(\SourceAlphabet, W_{\underlineS})$ denote this pair of sources. $S_{j}$ is observed
by encoder $j$ that has access to input $j$ of a two user discrete memoryless multiple
access channel (MAC) that is used without feedback. Let $\MACInputAlphabet_{1}$,
$\MACInputAlphabet_{2}$ be the finite input alphabet sets and $\MACOutputAlphabet$ the
finite output alphabet set of MAC. Let $W_{Y|X_{1}X_{2}}(y|x_{1},x_{2})$ denote MAC transition probabilities. We refer to this as MAC
$(\underlineMACInputAlphabet,\MACOutputAlphabet,W_{Y|\underlineX})$. The
objective of the decoder is to compute $S_{1}\oplus S_{2}$. In this
article, we provide a characterization of a sufficient condition for computing
$S_{1}\oplus S_{2}$ with arbitrary small probability of error. The relevant notions are
made
precise in the following definitions.

\begin{definition}
\label{Defn:ComputationCode}
A computation code $(n,\underlinee,d)$ for computing sum of sources $(\SourceAlphabet,
W_{\underlineS})$ over the MAC
$(\underlineMACInputAlphabet,\MACOutputAlphabet,W_{Y|\underlineX} )$
consists of (i) two encoder maps $e_{j}:
\SourceAlphabet^{n} \rightarrow \MACInputAlphabet_{j}^{n}:j=1,2$ and (ii) a decoder
map $d:\MACOutputAlphabet^{n} \rightarrow \ReconstructionAlphabet^{n}$.
\end{definition}

\begin{definition}
 \label{Defn:AverageErrorProbability}
 The average error probability $\bar{\xi}(\underlinee,d)$ of a computation code
$(n,\underlinee,d)$
is
 \begin{eqnarray}
  \label{Eqn:AverageErrorProbability}
   \sum_{\substack{\underlines\in \underlineSourceAlphabet^{n}}} \sum_{\substack{y^{n} :
d(y^{n})
\neq\\
s_{1}^{n}\oplus s_{2}^{n}}}\!\!\!
W_{Y^{n}|\underlineX^{n}}(y^{n}|\underlinee(\underlines^{n}))W_{\underlineS^{n}}
(\underlines^ { n } ) .\nonumber
 \end{eqnarray}
\end{definition}
\begin{definition}
 The sum of sources $(\SourceAlphabet,W_{\underlineS})$ is computable over MAC
$(\underlineMACInputAlphabet,\MACOutputAlphabet,W_{Y|\underlineX} )$ if
for all $\eta > 0$, there exists an $N(\eta) \in \naturals$ such that for all $n
> N(\eta)$, there exists an $(n,\underlinee^{(n)},d^{(n)})$ computation code such that
$\bar{\xi}(\underlinee^{(n)},d^{(n)}) \leq \eta$.
\end{definition}

The main objective in this article is to provide a sufficient
condition for computability of sum of sources over a MAC. 

\subsection{Linear Computation Coding}
\label{SubSec:LinearComputationCoding}
We describe the technique of LCC in a simple setting and highlight the key aspects.
Consider binary sources and a binary additive MAC, i.e., $\SourceAlphabet =
\MACInputAlphabet_{1} = \MACInputAlphabet_{2} = \left\{  0,1\right\}$ and
$Y=X_{1}\oplus X_{2} \oplus N$, where $N$ is independent of the inputs and
$P(N=1)=q$. Furthermore assume sources are symmetric, uniform, i.e., $P(\underlineS=(0,0))=\frac{1-p}{2}=P(\underlineS=(1,1))$ and
$P(\underlineS=(0,1))=P(\underlineS=(1,0))=\frac{p}{2}$ such that $h_{b}(p) < 1-h_{b}(q)$.

By employing a KM source code, the two message indices at rate
$h_{b}(p)$ can be employed to decode $S_{1}\oplus S_{2}$. Let $h \in
\SourceAlphabet^{k \times n}$ denote a parity check matrix for the KM source
code, with $\frac{k}{n} $ arbitrarily close to $ h_{b}(p)$. Nazer and Gastpar observe that
the decoder only requires the sum $h(S_{1}^{n}\oplus S_{2}^{n})=h(S_{1}^{n})\oplus
h(S_{2}^{n})$ of message indices. If the map from message indices to channel code
is linear, then the decoder can infer $h(S_{1}^{n})\oplus h(S_{2}^{n}) $ by decoding the
codeword corresponding to sum of transmitted codewords. Since sum of transmitted
codewords passes through a BSC($q$),
they employ a capacity achieving linear code of rate arbitrarily
close to $1-h_{b}(q)$ with generator matrix $g \in \MACInputAlphabet_{1}^{l
\times n}$. Each encoder employs the same linear code and transmits $x_{j}^{n}\define
h(S_{j}^{n})g$. The decoder receives $Y^{n}$ and decodes as if the channel is a BSC($q$). It ends up decoding message corresponding to $x_{1}^{n}
\oplus x_{2}^{n}$ which was precisely what it was looking for.

We note that a separation based scheme will require the sum capacity of the MAC to be
greater than $2h_{b}(p)$ and hence LCC is more efficient. What are key aspects of LCC?
Note that (i) the channel code is designed for the $X_{1}\oplus X_{2}$ to $Y$ channel,
i.e., the BSC($q$) and (ii) both encoders employ the same                                                           
linear channel code, thereby ensuring their codes are closed under addition. This contains range of the sum of transmitted codewords to a rate $1-h_{b}(q)$. It is instructive to analyze the case when the two users are provided two
linear codes of rates $R_{1}$ and $R_{2}$ spanning disjoint subspaces. Since the range
of sum of transmitted codewords is $R_{1}+R_{2}$, the same decoding rule will impose
the constraint $R_{1}+R_{2} < 1-h_{b}(q)$ resulting in the constraint $2h_{b}(p) \leq 
1-h_{b}(q)$ which is strictly suboptimal. \textit{We conclude that the two users' channel
codes being closed under addition is crucial to the optimality of LCC for this problem.}
Furthermore, the coupling of (i) a linear map of KM message indices to the
channel code at the encoder and (ii) decoding of the sum of transmitted codewords, is
central to LCC.

In the following section, we make use of the above observations to propose a
generalization of LCC for computing sum of sources over an arbitrary MAC.

\section{Nested coset codes for computing sum of sources over a MAC}
\label{Sec:NestedCosetCodesForComputingTheSumOfSourcesOverAMAC}
In this section, we propose a technique for computing $S_{1}\oplus S_{2}$ over an
\textit{arbitrary} MAC using the ensemble of nested coset codes \cite{201108ISIT_PadPra},
and derive a set of sufficient conditions under which, sum of sources
$(\SourceAlphabet,W_{\underlineS})$ can be computed over a MAC
$(\underlineMACInputAlphabet,\MACOutputAlphabet,W_{Y|\underlineX}
)$.
Definitions \ref{Defn:CharacterizationOfTestChannelsForComputingSumOverFields} and
theorem \ref{Thm:AchievabilityUsingNestedCosetCodes} state these sufficient conditions.This is followed by examples that illustrate significance of theorem \ref{Thm:AchievabilityUsingNestedCosetCodes}.
\begin{definition}
 \label{Defn:CharacterizationOfTestChannelsForComputingSumOverFields}
 Let $\mathbb{D}(W_{Y|\underlineX})$ be collection of distributions
$p_{V_{1}V_{2}X_{1}X_{2}Y}$ defined over $\SourceAlphabet^{2} \times
\underlineMACInputAlphabet \times
\MACOutputAlphabet$ such that (i) $p_{V_{1}X_{1}V_{2}X_{2}} =
p_{V_{1}X_{1}}p_{V_{2}X_{2}}$, (ii)
$p_{Y|\underlineX \underlineV} = p_{Y|\underlineX}=W_{Y|\underlineX}$. For
$p_{\underlineV\underlineX Y} \in
\mathbb{D}(W_{Y|\underlineX})$, let $\alpha(p_{\underlineV \underlineX Y})$ be defined as
\begin{eqnarray}
 \label{Eqn:CharacterizationOfComputabilityOfSumOverFields}
 \left\{ R \geq 0 : R \leq \min\{ H(V_{1}),H(V_{2}) \} - H(V_{1}\oplus V_{2}|Y)
 \right\},\mbox{ and} \nonumber \\
 \label{Eqn:SupremumOfAchievableRates}
 \alpha(W_{Y|\underlineX}) \define \sup \cup_{p_{\underlineV \underlineX Y} \in
\mathbb{D}(W_{Y|\underlineX})} \alpha(p_{\underlineV \underlineX Y}).\nonumber
\end{eqnarray}
\end{definition}

\begin{thm}
 \label{Thm:AchievabilityUsingNestedCosetCodes}
 The sum of sources $(\SourceAlphabet,W_{\underlineS})$is
computable over a MAC $(\underlineMACInputAlphabet,\MACOutputAlphabet,W_{Y|\underlineX
} )$ if $H(S_{1} \oplus S_{2}) \leq \alpha(W_{Y|\underlineX})$.
\end{thm}
Before we provide a proof, we briefly state the coding strategy and indicate how we attain the rates promised above.

We begin with a description of the encoding rule. Encoder $j$ employs a KM source code to
compress the observed source. Let $M_{j}^{l}\define hS_{j}^{n}$ denote corresponding
message index, where $h \in \SourceAlphabet^{l \times n}$ is a KM parity check matrix of
rate $\frac{l}{n} \approx H(S_{1}\oplus S_{2})$. Each encoder is provided with a common
nested linear code taking values over $\SourceAlphabet$. The nested linear code is
described through a pair of generator matrices $g_{I} \in \SourceAlphabet^{k \times n}$
and $g_{O/I} \in \SourceAlphabet^{l \times n}$, where $g_{I}$ and $\left[
g_{I}^{T} ~~g_{O/I}^{T} \right]^{T}$ are the generator matrices of the inner (sparser)
code and complete (finer) codes respectively, where
\begin{eqnarray}
\label{Eqn:BinningRate}
\frac{k}{n}\overset{(a)}{\geq} 1-\frac{\min\left\{\substack{H(V_{1}),\\H(V_{2})} \right\}}{\log |\mathcal{S}|}~~,~~
\frac{k+l}{n} \overset{(b)}{\leq} 1-\frac{H(V_{1} \oplus V_{2})}{\log |\mathcal{S}|}.
\end{eqnarray}
Encoder $j$ picks a codeword in coset $\left( a^{k}g_{I}\oplus M_{j}^{l}g_{O/I}:a^{k} \in \SourceAlphabet^{k}\right)$ indexed by $M_{j}^{l}$ that is typical with respect to $p_{V_{j}}$. Based on this chosen codeword $X^{n}$ is generated according to $p_{X_{j}|V_{j}}$ and transmitted.

The decoder is provided with the same nested linear code. Having received $Y^{n}$ it lists all codewords that are jointly typical with $Y^{n}$ with respect to distribution $p_{V_{1}\oplus V_{2},Y}$. If it finds all such codewords in a unique coset, say $\left( a^{k}g_{I}\oplus m^{l}g_{O/I}:a^{k} \in \SourceAlphabet^{k}\right)$, then it declares $m^{l}$ to be the sum of KM message indices and employs KM decoder to decode the sum of sources. Otherwise, it declares an error.

We derive an upper bound on probability of error by
averaging the error probability over the ensemble of nested linear codes. For the purpose
of proof, we consider user codebooks to be cosets of nested linear codes.\footnote{This is
analogous to the use of cosets of a linear code to prove achievability of symmetric
capacity over point-to-point channels.} We average uniformly over the entire ensemble of
nested \textit{coset} codes. Lower bound (\ref{Eqn:BinningRate}(a)) ensures the encoders find a typical
codeword in the particular coset. Upper bound (\ref{Eqn:BinningRate}(b)) enables us derive an upper
bound on the probability of decoding error. From (\ref{Eqn:BinningRate}), it can be verified that if $H(S_{1}\oplus S_{2})\approx
\frac{l}{n} \leq\min\{ H(V_{1}),H(V_{2}) \} - H(V_{1}\oplus V_{2}|Y)$ then the
decoder can reconstruct the sum of sources with arbitrarily small probability of error.

How does nesting of linear codes enable attain non-uniform distributions?\footnote{Note that linear codes only achieve mutual information with respect to uniform input distributions.} As against to a linear code, nesting of linear codes provides the encoder with a coset to choose the codeword from. The vectors in the coset being uniformly distributed, it contains at least one vector typical with respect to $p_{V_{j}}$ with high probability, if the coset is of rate at least $1-\frac{H(V_{j})}{\log |\mathcal{S}|}$. By choosing such a vector, the encoder induces a non-uniform distribution on the input space. Therefore, constraint (\ref{Eqn:BinningRate}(a)) enables achieve non-uniform input distributions.

Since the codebooks employed by the encoders are uniformly and independently distributed cosets of a common random linear code, the sum of transmitted codewords also lies in a codebook that is a uniformly distributed coset of the same linear code. Any vector in this codebook is uniformly distributed over it's entire range. Therefore, a vector in this codebook other than the legitimate sum of transmitted codewords is jointly typical with the received vector with probability at most $|\mathcal{S}|^{n(H(V_{1}\oplus V_{2}|Y)-1)}$.\footnote{Here, the logarithm is taken with respect to base $|\mathcal{S}|$.} Employing a union bound, it can be argued that the probability of decoding error decays exponentially if (\ref{Eqn:BinningRate}(b)) holds.

Since the ensemble of codebooks contain statistically dependent codewords and moreover user codebooks are closely related, deriving an upper bound on the probability of error involves new elements. The informed reader will
recognize that in particular, deriving an upper bound on the probability of decoding error will involve proving
statistical independence of the pair of cosets indexed by KM indices $(M_{1}^{l},M_{2}^{l})$                  
and any codeword in a coset corresponding to $\hatm^{l} \neq M_{1}^{l} \oplus M_{2}^{l}$.
The statistical dependence of the codebooks results in new elements to the proof. The
reader is encouraged to peruse the same in the following.
\begin{proof}
Given $\eta > 0$, our goal is to identify a computation code $(n,\underlinee,d)$ such
that $P(d(Y^{n})\neq S_{1}^{n}\oplus S_{2}^{n}) \leq \eta$ for all sufficiently large $n \in \mathbb{N}$.
The source sequences are mapped to channel input codewords in two stages. In the first stage, a
distributed source code proposed by K\"orner and Marton \cite{197903TIT_KorMar} is
employed to map $n$-length source sequences to message indices that takes values over
$\SourceAlphabet^{l}$. The second stage maps these indices to channel input codewords. We
begin by stating the main findings of \cite{197903TIT_KorMar} on which our first stage
relies.
\begin{lemma}
\label{Lem:KornerMartonResult}
 Given a pair of $(\SourceAlphabet,W_{\underlineS})$ of information sources and $\eta
>
0$, there exists an $N(\eta) \in \mathbb{N}$ such that for every $n \in \mathbb{N}$,
there exists a parity check matrix $h \in \SourceAlphabet^{l(n) \times n}$ and a map
$r:\SourceAlphabet^{l(n)} \rightarrow \SourceAlphabet^{n}$ such that (i)
$\frac{l(n)}{n} \leq H(S_{1} \oplus S_{2}) + \frac{\eta}{2}$, and (ii)
$P(r(hS_{1}^{n}\oplus hS_{2}^{n})
\neq S_{1}^{n} \oplus S_{2}^{n}) \leq \frac{\eta}{2}$.
\end{lemma}

Given $\eta >0$, let $h \in \SourceAlphabet^{l \times n}$ be a parity check matrix
that satisfies (i) and (ii) in lemma \ref{Lem:KornerMartonResult}. Let $M_{j}^{l} \define
hS_{j}^{n}:j=1,2$ be the message indices output by the source encoder. In the
second stage, we identify maps $\mu_{j} : \SourceAlphabet^{l} \rightarrow
\MACInputAlphabet_{j}^{n}:j=1,2$ that maps these message indices to channel input
codewords. The
encoder $e_{j}:\SourceAlphabet^{n} \rightarrow \MACInputAlphabet_{j}^{n}$ of the
computation code is therefore defined as $e_{j}(S_{j}^{n}) \define
\mu_{j}(hS_{j}^{n})$. The second stage of the encoding is based on nested coset codes.
We begin with a brief review of nested coset codes.

An $(n,k)$ coset is a collection of vectors in $\fieldq^{n}$ obtained by adding a constant
bias vector to a $k-$dimensional subspace of $\fieldq^{n}$. If $\lambda_{O}
\subseteq \fieldq^{n}$ and $\lambda_{I}\subseteq \lambda_{O}$ are $(n,k+l)$ and
$(n,k)$ coset codes respectively, then $q^{l}$ cosets $\lambda_{O}/\lambda_{I}$ that
partition $\lambda_{O}$ is a nested coset code.

A couple of remarks are in order. An $(n,k)$ coset code is specified by a bias vector
$b^{n} \in \fieldq^{n}$ and generator matrices $g \in \fieldq^{k \times
n}$. If $\lambda_{O} \subseteq \fieldq^{n}$ and $\lambda_{I} \subseteq \lambda_{O}$ are
$(n,k+l)$ and $(n,k)$ coset codes respectively, then there exists a bias vector $b^{n} \in
\fieldq^{n}$ and generator matrices $g_{I} \in \fieldq^{k \times
n}$ and $g_{O} = \left[ \begin{array}{c} g_{I} \\ g_{O/I} \end{array} \right ] \in
\fieldq^{(k+l)\times n}$, such that $b^{n}$, $g_{I}$ specify
$\lambda_{I}$ and $b^{n}$, $g_{O}$ specify $\lambda_{O}$. Therefore, a nested coset code
is
specified by a bias vector $b^{n}$ and any two of the three generator matrices $g_{I}$,
$g_{O/I}$ and $g_{O}$. We refer to this as nested coset code
$(n,k,l,g_{I},g_{O/I},b^{n})$.

We now specify the encoding rule. Encoder $j$ is provided a nested
coset code $(n,k,l,g_{I},g_{O/I},b_{j}^{n})$ denoted $\lambda_{Oj}/\lambda_{I}$ taking
values over the finite field $\SourceAlphabet$. Let
$v_{j}^{n}(a^{k},m_{j}^{l}) \define a^{k}g_{I}\oplus m_{j}^{l}g_{O/I}\oplus b_{j}^{n}$
denote a
generic codeword in $\lambda_{Oj}/\lambda_{I}$ and $c_{j}(m^{l}_{j}) \define
(v_{j}^{n}(a^{k},m_{j}^{l}):a^{k} \in \SourceAlphabet^{k})$ denote coset corresponding to
message $m_{j}^{l}$. The message index $M_{j}^{l} = hS_{j}^{n}$ put out by the
source encoder is used to index coset $c_{j}(M_{j}^{l})$. Encoder $j$ looks for a codeword
in coset $c(M_{j}^{l})$ that is typical according
to $p_{V_{j}}$. If it finds at least one such codeword, one of them, say
$v_{j}^{n}(a^{k},M_{j}^{l})$ is chosen uniformly at random. $\mu_{j}(M_{j}^{l})$ is
generated according $p_{X^{n}|V^{n}}(\cdot|v_{j}^{n}(a^{k},M_{j}^{l})) =
\prod_{t=1}^{n}p_{X_{j}|V_{j}}(\cdot|(v_{j}^{n}(a^{k},M_{j}^{l}))_{t})$ and
$\mu_{j}(M_{j}^{l})$ is transmitted. Otherwise, an error is
declared.

We now specify the decoding rule. The decoder is provided with the nested coset code
$(n,k,l,g_{I},g_{O/I},b^{n})$ denoted $\lambda_{O}/\lambda_{I}$, where
$b^{n}=b_{1}^{n}\oplus b_{2}^{n}$. We employ notation similar to that specified for the
encoder. In particular, let
$v^{n}(a^{k},m^{l}) \define a^{k}g_{I}\oplus m^{l}g_{O/I}\oplus b^{n}$ denote a
generic codeword and $c(m^{l}) \define
(v^{n}(a^{k},m^{l}):a^{k} \in \SourceAlphabet^{k})$ denote a generic coset in
$\lambda_{O}/\lambda_{I}$ respectively. Decoder receives $Y^{n}$ and declares error if
$Y^{n} \notin T_{\frac{\eta_{1}}{2}}(p_{Y})$. Else, it lists
all codewords $v^{n}(a^{k},m^{l}) \in \lambda_{O}$ such that
$(v^{n}(a^{k},m^{l}),Y^{n})
\in T_{\eta_{1}}^{n}(p_{V_{1}\oplus  V_{2},Y})$. If it finds all such codewords in a
unique
coset say $c(m^{l})$ of
$\lambda_{O}/\lambda_{I}$, then it declares $r(\hat{m}^{l})$ to be the decoded sum of
sources, where $r:\SourceAlphabet^{l} \rightarrow \SourceAlphabet^{n}$ is as
specified in lemma \ref{Lem:KornerMartonResult}. Otherwise, it declares an error.

As is typical in information theory, we derive an upper bound on probability of error by
averaging the error probability over the ensemble of nested coset codes. We average over
the ensemble of nested coset codes by letting the bias vectors
$B_{j}^{n}:j=1,2$ and generator matrices $G_{I}, G_{O/I}$ mutually independent and
uniformly distributed over their respective range spaces. Let
$\Lambda_{Oj}/\Lambda_{I}:j=1,2$ and $\Lambda_{O}/\Lambda_{I}$ denote the random
nested coset codes $(n,k,l,G_{I},G_{O/I},B_{j}^{n}):j=1,2$ and
$(n,k,l,G_{I},G_{O/I},B^{n})$ respectively, where $B^{n}=B_{1}^{n}\oplus B_{2}^{n}$. For
$a^{k}
\in \SourceAlphabet^{k}$,
$m^{l} \in \SourceAlphabet^{l}$, let $V_{j}^{n}(a^{k},m_{j}^{l}):j=1,2$,
$V^{n}(a^{k},m^{l})$ denote corresponding random codewords in
$\Lambda_{Oj}/\Lambda_{I}:j=1,2$ and $\Lambda_{O}/\Lambda_{I}$ respectively. Let
$C_{j}(m^{l}_{j}) \define
(V_{j}^{n}(a^{k},m_{j}^{l}):a^{k} \in \SourceAlphabet^{k})$ and $C(m^{l})\define
(V^{n}(a^{k},m^{l}):a^{k} \in \SourceAlphabet^{k})$ denote random cosets in
$\Lambda_{Oj}/\Lambda_{I}:j=1,2$ and $\Lambda_{O}/\Lambda_{I}$ corresponding to
message $m_{j}^{l}:j=1,2$ and $m^{l}$ respectively. We now
analyze error events and
upper bound probability of error.

We begin by characterizing error events at encoder. If $\phi(m_{j}^{l})
\define
\sum_{a^{k} \in
\SourceAlphabet^{k}}
1_{\{\left(V_{j}^{n}(a^{k},m_{j}^{l})\right) \in
T_{\eta_{2}}^{n}(p_{V_{j}})\}}$ and $\epsilon_{j1} \define \{
\phi(hS_{j}^{n}) =0
\}$, then $\epsilon_{j1}$ is the error event at encoder $j$.
An
upper bound on $P(\epsilon_{j1})$ can be derived by following the
arguments in  [Proof of Theorem1]\cite{201108ISIT_PadPra}. Findings in
\cite{201108ISIT_PadPra} imply existence of $N_{j2} \in \naturals$ such
that $\forall n\geq N_{j2}$, $P(\epsilon_{j1}) \leq \frac{\eta}{8}$
if $\frac{k}{n} > 1-\frac{H(V_{j})}{\log |\mathcal{S}|}$.

The error event at the decoder is $\epsilon_{2} \cup \epsilon_{3}$, where
$\epsilon_{2}\define\{ Y^{n} \notin
T_{\frac{\eta_{1}}{2}}^{n}(p_{Y}) \}$ and
\begin{equation}
\label{Eqn:DefnOfEpsilon3}
\epsilon_{3}\define\underset{\substack{ m^{l} \neq \\hS_{1}^{n}\oplus
hS_{2}^{n} }  }{\bigcup}\underset{a^{k} \in
\SourceAlphabet^{k}}{\bigcup}\left\{\substack{\left(V^{n}(a^{k},m^{l}), Y^{n}\right) \in
T_{\eta_{1}}^{n}(p_{V_{1}\oplus  V_{2},Y})}\right\}.\nonumber
\end{equation}
In order to upper bound
$P(\epsilon_{2})$
by conditional frequency typicality, it suffices to upper bound
$P((\underlinee(\underlineS^{n})) \notin T_{\frac{\eta_{1}}{4}}(p_{\underlineX}))$.
Note
that (i) independence of $(V_{j},X_{j}):j=1,2$ implies the Markov chain
$X_{1}-V_{1}-V_{2}-X_{2}$, and (ii) the chosen codeword $V_{j}^{n}(a^{k},M_{j}^{l})$ and
the transmitted vector $e_{j}(S_{j}^{n})= \mu_{j}(M_{j}^{l})$ are jointly typical with
high probability as a consequence of conditional generation of the latter. By the Markov
lemma \cite{201201NIT_ElgKim}, it suffices to prove $V_{j}^{n}(a^{k},M_{j}^{l}):j=1,2$ are jointly
typical. If the codewords were chosen independently at random according to
$\prod_{t=1}^{n}p_{V_{j}}$, this would fall out as a consequence of uniformly sampling
from the typical set \cite[]{201201NIT_ElgKim}. However, the generation of nested coset
code is different, and the proof of this involves an alternate route. An analogous proof
of the Markov lemma is provided in \cite{201301arXivMACDSTx_PadPra} and omitted here in the interest of brevity.

It remains to upper bound $P((\epsilon_{11}\cup \epsilon_{21}\cup
\epsilon_{2})^{c}\cap\epsilon_{3})$. In appendix
\ref{Sec:AnUpperBoundOnProbabilityofEpsilon3}, we prove that if $\frac{k+l}{n} < 1 - H
(V_{1} \oplus V_{2}|Y)$, there exists $N_{4}(\eta) \in \mathbb{N}$ such that
$\forall n \geq N_{4}$, $P(\epsilon_{3}) \leq \frac{\eta}{8}$. Combining the bounds
$\frac{k}{n} > 1-H(V_{j})$ and $\frac{k+l}{n} < 1 - H
(V_{1} \oplus V_{2}|Y)$, we note that $\frac{l}{n} < \min \left\{ H(V_{1}),H(V_{2})
\right\}-H(V_{1} \oplus V_{2}|Y)$, then the sum of message indices $h(S_{1}^{n} \oplus
S_{2}^{n})$ can be reconstructed at the decoder. This concludes proof of
achievability.

The informed reader will
recognize that deriving an upper bound on $P(\epsilon_{3})$ will involve proving
statistical independence of the pair $(C_{j}(hS_{j}^{n}):j=1,2)$ of cosets
and any codeword $V^{n}(\hata^{k},\hatm^{l})$ corresponding to a competing sum of
messages $\hatm^{l} \neq h(S_{1}^{n} \oplus S_{2}^{n})$. This is considerably simple for
a coding technique based on classical unstructured codes wherein codebooks and codewords
in every codebook are independent. The coding technique proposed herein involves
correlated codebooks and codewords resulting in new elements to the proof. The reader is
encouraged to peruse details of this element presented in appendix
\ref{Sec:AnUpperBoundOnProbabilityofEpsilon3}.
\end{proof}

It can be verified that the rate region presented in theorem \ref{Thm:AchievabilityUsingNestedCosetCodes} subsumes that presented in \cite[Theorem1, Corollary 2]{200710TIT_NazGas}. This follows by substituting a uniform distribution for $V_{1},V_{2}$. Therefore examples presented in \cite{200710TIT_NazGas} carry over as examples of rates achievable using nested coset codes. One might visualize a generalization of LCC for arbitrary MAC through the modulo-lattice transformation (MLT) \cite[Section IV]{201207ISIT_HaiKocEre}. Since the map for KM source code message indices to the channel code has to be linear, the virtual input alphabets of the transformed channel are restricted to be source alphabets as in definition \ref{Defn:CharacterizationOfTestChannelsForComputingSumOverFields}. It can now be verified that any virtual channel, specified through maps from (i) virtual to actual inputs, (ii) output to the estimate of the linear combination, identifies a corresponding test channel in $\mathbb{D}(W_{Y|\underline{X}})$. Hence, the technique proposed herein subsumes MLT. Moreover, while MLT is restricted to employing uniform distributions over the auxiliary inputs, nested coset codes can induce arbitrary distributions.

We now present a sample of examples to illustrate significance of theorem \ref{Thm:AchievabilityUsingNestedCosetCodes}. As was noted in \cite[Example 4]{200710TIT_NazGas} a uniform distribution induced by a linear code maybe suboptimal even for computing functions over a MAC with a structural match. The following example, closely related to the former, demonstrates the ability of nested coset codes to achieve a nonuniform distribution and thus exploit the structural match better.
\begin{example}
 \label{Ex:PentaAdditiveMACWithTwoNoisyAndThreeNoiselessSymbols}
Let $S_{1}$ and $S_{2}$ be a pair of independent and uniformly distributed sources taking
values over the field $\mathcal{F}_{5}$ of five elements. The decoder wishes to
reconstruct $S_{1} \oplus_{5} S_{2}$. The two user MAC channel input alphabets
$\MACInputAlphabet_{1}=\MACInputAlphabet_{2}=\mathcal{F}_{5}$ and output alphabet
$\MACOutputAlphabet=\left\{ 0,2,4 \right\}$. The output $Y$ is obtained by passing
$W=X_{1}\oplus_{5}X_{2}$ through an asymmetric channel whose transition probabilities are
given by $p_{Y|W}(y|1)=p_{Y|W}(y|3)=\frac{1}{3}$ for each $y \in \MACOutputAlphabet$ and
$p_{Y|W}(0|0)=p_{Y|W}(2|2)=p_{Y|W}(4|4)=1$. Let the number of source digits output per
channel use be $\lambda$. We wish to compute the range of values of $\lambda$ for which
the decoder can reconstruct the sum of sources. This is termed as computation rate in
\cite{200710TIT_NazGas}.

It can be verified that the decoder can reconstruct $S_{1} \oplus_{5} S_{2}$ using the technique of LCC if $\lambda\leq \frac{3}{5}\frac{\log_{2}(3)}{\log_{2}5}=0.4096$. A separation based scheme enables the decoder reconstruct the sum if $\lambda \leq \frac{1}{2}\frac{\log_{2}(3)}{\log_{2}(5)}=0.3413$. We now explore the use of nested coset codes. It maybe verified that pmf 
\begin{equation}
 \label{Eqn:PMFForNestedCosetCodes}
p_{\underlineV\underlineX Y}(\underline{v},\underline{x},x_{1}\oplus_{5}x_{2})=
\begin{cases}
 \frac{1}{4} \substack{\mbox{ if }v_{1}=x_{1},v_{2}=x_{2}\\ \mbox{ and }v_{1},v_{2} \in \left\{ 0,2 \right\}}\\
0\mbox{ otherwise }.
\end{cases}
\end{equation}
defined on $\mathcal{F}_{5}  \times \mathcal{F}_{5}$ satisfies (i),(ii) of definition \ref{Defn:CharacterizationOfTestChannelsForComputingSumOverFields} and moreover
$\alpha(p_{\underlineV \underlineX Y}) = \left\{ R \geq 0 : R \leq 1 \right\}$.
Thus nested coset codes enable reconstructing $S_{1} \oplus_{5} S_{2}$ at the decoder if $\lambda\leq \frac{1}{\log_{2}5}=.43067$.
\end{example}
The above example illustrates the need for nesting codes in order to achieve nonuniform
distributions. However, for the above example, a suitable modification of LCC is optimal.
Instead of building codes over $\mathcal{F}_{5}$, let each user employ the linear code of
rate $1$\footnote{This would be the set of all binary $n-$length vectors} built on
$\BinaryField$. The map $\BinaryField \rightarrow \MACInputAlphabet_{j}:j=1,2$ defined as
$0 \rightarrow 0$ and $1 \rightarrow 2$ induces a code over $\mathcal{F}_{5}$ and it can
be verified that LCC achieves the rate achievable using nested coset codes. However, the
following example precludes such a modification of LCC.

\begin{example}
 \label{Ex:PentaAdditiveMACWithTwoUselessAndThreeNoisySymbols}
The source is assumed to be the same as in example
\ref{Ex:PentaAdditiveMACWithTwoNoisyAndThreeNoiselessSymbols}. The two user MAC input and
output alphabets are also assumed the same, i.e.,
$\MACInputAlphabet_{1}=\MACInputAlphabet_{2}=\mathcal{F}_{5}$ and output alphabet
$\MACOutputAlphabet=\left\{ 0,2,4 \right\}$. The output $Y$ is obtained by passing
$W=X_{1}\oplus_{5}X_{2}$ through an asymmetric channel whose transition probabilities are
given by $p_{Y|W}(y|1)=p_{Y|W}(y|3)=\frac{1}{3}$ for each $y \in \MACOutputAlphabet$ and
$p_{Y|W}(0|0)=p_{Y|W}(2|2)=p_{Y|W}(4|4)=0.90,
p_{Y|W}(2|0)=p_{Y|W}(4|0)=p_{Y|W}(0|2)=p_{Y|W}(4|2)=p_{Y|W}(0|4)=p_{Y|W}(2|4)=0.05$.

The technique of LCC builds a linear code over $\mathcal{F}_{5}$. It can be verified that
the symmetric capacity for the $X_{1}\oplus_{5}X_{2}(=W) - Y$ channel is $0.6096$ and
therefore LCC enables decoder reconstruct the sum if $\lambda \leq
\frac{0.6096}{\log_{2}5}=0.2625$. A separation based scheme necessitates communicating
each of the sources to the decoder and this can be done only if $\lambda \leq
\frac{1}{2}\frac{\log_{2}3}{\log_{2}5}=0.3413$. The achievable rate region of the test
channel in (\ref{Eqn:PMFForNestedCosetCodes}) is $\alpha(p_{\underlineV \underlineX Y}) =
\left\{ R \geq 0 : R \leq 0.91168 \right\}$ and therefore nested coset codes enable
decoder reconstruct the sum if $\lambda \leq \frac{0.91168}{\log_{2}5}=0.3926$.
\end{example}

\begin{example}
\label{Ex:TernarySourcesAndDecoderIntersetedInParityOfThePair}
 Let $S_{1}$ and $S_{2}$ be independent sources distributed uniformly over
$\left\{0,1,2\right\}$. The input alphabets
$\MACInputAlphabet_{1}=\MACInputAlphabet_{2}=\mathcal{F}_{3}$ is the ternary field and the
output alphabet $\MACOutputAlphabet=\mathcal{F}_{2}$ is the binary field. Let
$W=1_{\left\{ X_{1} \neq X_{2}\right\}}$ and output $Y$ is obtained by passing $W$ through
a BSC with crossover probability $0.1$. The decoder is interested in reconstructing $W$.
As noted in \cite[Example 8]{200710TIT_NazGas}, $W$ is $0$ if an only if $S_{1} \oplus_{3}
2S_{2} =0$. Therefore, it suffices for the decoder to reconstruct $S_{1} \oplus_{3}
2S_{2}$. Following the arguments in proof of theorem
\ref{Thm:AchievabilityUsingNestedCosetCodes} it can be proved that $S_{1} \oplus_{3}
2S_{2}$ can be reconstructed using nested coset codes if there exists a pmf
$p_{\underlineV \underlineX Y} \in \mathbb{D}(W_{Y|\underlineX})$ such that $H(S_{1}
\oplus_{3} 2S_{2}) \leq \min\{ H(V_{1}),H(V_{2}) \} - H(V_{1}\oplus_{3}2 V_{2}|Y)$. It can
be verified that for pmf $p_{\underlineV \underlineX Y}$ wherein $V_{1},V_{2}$ are
independently and uniformly distributed over $\mathcal{F}_{3}$, $X_{1}=V_{1}$,
$X_{2}=V_{2}$, the achievable rate region is $\alpha(p_{\underlineV \underlineX
Y})=\left\{ R: R\leq  0.4790\right\}$. The computation rate achievable using SCC and
separation technique are $0.194$ and $0.168$ respectively. The computation rate achievable
using nested coset codes is $\frac{0.4790}{\log_{2}3}=0.3022$.
\end{example}

\begin{example}
\label{Example:AdditiveTweakedToMakeNonAdditive}
 Let $S_{1}$ and $S_{2}$ be independent and uniformly distributed binary sources and the
decoder is interested in reconstructing the binary sum. The MAC is binary, i.e.
$\MACInputAlphabet_{1}=\MACInputAlphabet_{2}=\MACOutputAlphabet=\BinaryField$ with
transition probabilities $P(Y=0|X_{1}=x_{1},X_{2}=x_{2})=0.1$ if $x_{1}\neq x_{2}$,
$P(Y=0|X_{1}=X_{2}=0)=0.8$ and $P(Y=0|X_{1}=X_{2}=1)=0.9$. It can be easily verified that
the channel is not linear, i.e., $\underlineX-X_{1}\oplus X_{2}-Y$ is \textit{NOT} a
Markov chain. This restricts current known techniques to either separation based coding or
SCC \cite[Section V]{200710TIT_NazGas}. SCC yields a computation rate of $0.3291$. The
achievable rate region for the test channel $p_{\underlineV \underlineX Y}$ where in
$V_{1}$ and $V_{2}$ are independent and uniformly distributed binary sources,
$X_{1}=V_{1},X_{2}=V_{2}$ is given by $\left\{ R: R\leq 0.4648 \right\}$.
\end{example}
We conclude by recognizing that example \ref{Example:AdditiveTweakedToMakeNonAdditive} is indeed a family of examples. As long as the MAC is close to additive we can expect nested coset codes to outperform separation and SCC.
\section{General technique for computing sum of sources over a MAC}
\label{Sec:GeneralTechnique}
In this section, we propose a general technique for computing sum of sources over a MAC
that subsumes separation and computation. The architecture of the code we propose is
built on the principle that techniques based on structured coding are not in lieu of
their counterparts based on unstructured coding. Indeed, the KM technique is
outperformed by the Berger-Tung \cite{Berger-MSC} strategy for a class of source distributions. A general
strategy must therefore incorporate both.

We take the approach of Ahlswede and Han \cite[Section VI]{198305TIT_AhlHan}, where in a
two layer source code is proposed. Each
source encoder $j$ generates two message indices $M_{j1},M_{j2}$. $M_{j1}$ is an index
to a Berger-Tung source code and $M_{j2}$ is an index to a KM source
code. The source decoder therefore needs $M_{11},M_{21}$ and $M_{12}\oplus M_{22}$ to
reconstruct the quantizations and thus the sum of sources. We propose a two layer MAC
channel code that is compatible with the above source code. The first layer of this code
is a standard MAC channel code based on unstructured codes. The messages input to this layer are communicated
as is to the decoder. The second layer employs nested coset codes and is identical to the
one proposed in theorem \ref{Thm:AchievabilityUsingNestedCosetCodes}. A function of the
codewords selected from each layer is input to the channel. The decoder decodes a triple
- the pair of codewords selected from the first layer and a sum of codewords selected
from the second layer - and thus reconstructs the required messages. The following
characterization specifies rates of layers 1 and 2 separately and therefore differs
slightly from \cite[Theorem 10]{198305TIT_AhlHan}.
\begin{definition}
 \label{Defn:TestChannelsAhlswedeHanCharacterization}
Let $\mathbb{D}_{\mbox{\tiny{AH}}}(W_{\underlineS})$ be collection of
distributions $p_{T_{1}T_{2}S_{1}S_{2}}$ defined over
$\mathcal{T}_{1}\times \mathcal{T}_{2}\times \SourceAlphabet^{2}$ such
that (a) $\mathcal{T}_{1},\mathcal{T}_{2}$ are finite sets, (b) $p_{S_{1}S_{2}}=W_{S}$, (c) $T_{1}-S_{1}-S_{2}-T_{2}$ is a
Markov chain. For $p_{\underlineT \underlineS} \in
\mathbb{D}_{\mbox{\tiny{AH}}}(W_{\underlineS})$, let
\begin{eqnarray}
 \label{Eqn:AhlswedeHanReateRegionForSpecificTestChannel}
\beta_{S}(p_{\underlineT \underlineS
})\!\define\! \!  \left\{
\begin{array}{l}(R_{11},R_{12},R_{2}) \in \reals^{3} : R_{11} \geq I(T_{1};S_{1}|T_{2}),\\
R_{12} \geq I(T_{2};S_{2}|T_{1}),R_{2} \geq H(S_{1} \oplus
S_{2}|\underlineT),\\ R_{11}+R_{12} \geq I(\underlineT
; \underlineS)\end{array}\right\}.\nonumber
\end{eqnarray}
Let $\beta_{S}(W_{\underlineS})$ denote convex closure of the union $\beta_{S}(p_{\underlineT \underlineS
})$ over $p_{\underlineT
\underlineS} \in
\mathbb{D}_{\mbox{\tiny{AH}}}(W_{\underlineS})$
\end{definition}
We now characterize achievable rate region for communicating these indices over a MAC. We
begin with a definition of test channels and the corresponding rate region.

\begin{definition}
 \label{Defn:ChannelCodingTestChannelsForAhlswedeHanCharacterization}
 Let $\mathbb{D}_{\mbox{\tiny{G}}}$ be collection of distributions
$p_{U_{1}U_{2}V_{1}V_{2}X_{1}X_{2}Y}$ defined on
$\mathcal{U}_{1} \times \mathcal{U}_{2} \times \SourceAlphabet \times \SourceAlphabet
\times \MACInputAlphabet_{1} \times \MACInputAlphabet_{2} \times \MACOutputAlphabet$ such
that (i) $p_{\underlineU \underlineV \underlineX} =
p_{U_{1}V_{1}X_{1}}p_{U_{2}V_{2}X_{2}}$,
(ii) $p_{Y|\underlineX \underlineU \underlineV} = p_{Y|\underlineX}=W_{Y|\underlineX}$.
For
$p_{\underlineU \underlineV \underlineX Y} \in
\mathbb{D}_{\mbox{\tiny{G}}}$, let $\beta_{C}(p_{\underlineU \underlineV \underlineX Y})$ be
defined as
\begin{eqnarray}
  \label{Eqn:ChannelCodingAhlswedeHanReateRegionForSpecificTestChannel}
\left\{
\begin{array}{l}\scriptstyle(R_{11},R_{12},R_{2}) \in \reals^{3} : 0\leq R_{11} \leq
I(U_{1};Y,U_{2},V_{1}\oplus V_{2}),\\\scriptstyle
0\leq R_{12} \leq
I(U_{2};Y,U_{1},V_{1}\oplus V_{2}), R_{11}+R_{12} \leq
I(\underlineU;Y,V_{1}\oplus V_{2})\\ \scriptstyle R_{2} \leq \mathscr{H}_{\min}(V|U)-H(V_{1}\oplus V_{2}|Y,\underlineU)\\\scriptstyle
R_{11}+R_{2} \leq \mathscr{H}_{\min}(V|U)+H(U_{1})-H(V_{1}\oplus
V_{2},U_{1}|Y,U_{2})\\\scriptstyle
R_{12}+R_{2} \leq \mathscr{H}_{\min}(V|U)+H(U_{2})-H(V_{1}\oplus
V_{2},U_{2}|Y,U_{1})\\\scriptstyle
R_{11}+R_{12}+R_{2} \leq \mathscr{H}_{\min}(V|U)+H(U_{1})+H(U_{2})-H(V_{1}\oplus
V_{2},\underlineU|Y)
\end{array}\right\}.\nonumber
\end{eqnarray}
where $\mathscr{H}_{\min}(V|U) \define \min\{ H(V_{1}|U_{1}),H(V_{2}|U_{2})\}$ and define $\beta_{C}(W_{Y|\underlineX})$ as the convex closure of the union $\beta_{C}(p_{\underlineU
\underlineV \underlineX Y})$ over $p_{\underlineU \underlineV \underlineX Y} \in \mathbb{D}_{\mbox{\tiny{G}}}(W_{Y|\underlineX})$.
\end{definition}

\begin{thm}
 \label{Thm:JointSeparationComputationTechniqueBasedOnAhlswedeHan}
The sum of sources $(\SourceAlphabet,W_{S})$ is computable over MAC
$(\underlineMACInputAlphabet,\MACOutputAlphabet,W_{Y|\underlineX
} )$ if $\beta_{S}(W_{\underlineS}) \cap \beta_{C}(W_{Y|\underlineX})\neq \phi$.
\end{thm}

\begin{remark}
 \label{Rem:UniversalTechniqueSubsumesSeparationAndComputation}
 It is immediate that the general strategy subsumes separation and computation based
techniques. Indeed, substituting $\underlineT,\underlineU$ to be degenerate yields the
conditions provided in theorem \ref{Thm:AchievabilityUsingNestedCosetCodes}. Substituting
$\underlineV$ to be degenerate yields separation based technique.
\end{remark}
\appendices
\section{An upper bound on
$P(\epsilon_{3})$}
\label{Sec:AnUpperBoundOnProbabilityofEpsilon3}
In this appendix, we derive an upper bound on $P(\epsilon_{3})$. As is typical in proofs
of channel coding theorems, this step involves establishing statistical independence of
$C_{j}(hS_{j}^{n}):j=1,2$ and any codeword $V^{n}(a^{k},\hatm^{l})$ in a competing coset
$\hatm^{l} \neq hS_{1}^{n}\oplus hS_{2}^{n}$. We establish this in lemma
\ref{Lem:StatisticalIndependenceOfTransmittedCosetAndCompetingCodeword}.
We begin with the necessary spadework. The following lemmas holds for any $\fieldq$ and we
state it in this generality.
\begin{lemma}
 \label{Lem:UniformDistributionAndPairwiseIndependenceOfCodewordsInARandomCoset}
 Let $\fieldq$ be a finite field. Let $G_{I} \in \fieldq^{k \times n}$, $G_{O/I} \in
\fieldq^{l \times n}$, $B_{j}^{n} \in \fieldq^{n}:j=1,2$ be mutually independent and
uniformly distributed on their respective range spaces. Then the following hold.
\begin{enumerate}
 \item[(a)] $P(V^{n}(a^{k},m^{l})=v^{n})=\frac{1}{q^{n}}$ for any $a^{k}
\in\fieldq^{k}$, $m^{l} \in \fieldq^{l}$ and $v^{n} \in \fieldq^{n}$,
\item[(b)] $P(V_{j}^{n}(a_{j}^{k},m_{j}^{l})=v_{j}^{n}:j=1,2)=\frac{1}{q^{2n}}$ for
any $a_{j}^{k}
\in\fieldq^{k}$, $m_{j}^{l} \in \fieldq^{l}$ and $v_{j}^{n} \in
\fieldq^{n}:j=1,2$, and 
\item[(c)] $P\left(\substack{
V_{j}^{n}(0^{k},m_{j}^{l})=v_{j,0^{k}}^{n}:j=1,2,\\V^{n}(0^{k},{
\hatm^{l}}
)=v^{n}}\right)=\frac{1}{q^{3n}}$ for any $\hatm^{l} \neq
m_{1}^{l} \oplus m_{2}^{l}$ and $v_{j,0^{k}}^{n}:j=1,2,$ and $v^{n}$.
\end{enumerate}
\end{lemma}
\begin{proof}
 The proof follows from a counting argument similar to that employed in
\cite[Remarks 1,2]{201108ISIT_PadPra}.\\
(a) For any $g_{I} \in \fieldq^{k \times n}$,
$g_{O / I} \in  \fieldq^{l \times n}$, $v^{n} \in \fieldq^{n}$, there exists a unique
$b^{n} \in\fieldq^{n}$ such that $a^{k}g_{I}\oplus m^{l}g_{O / I} \oplus  b^{n} = v^{n}$.
Since
$G_{I}$, $G_{O/I}$ and $B^{n}$ are mutually independent and uniformly distributed
$P(V^{n}(a^{k},m^{l})=v^{n})=\frac{q^{kn}q^{ln}}{q^{kn}q^{ln}q^{n}}=\frac{1}{
q^ { n } } $.\\
(b) We first note
$P(V_{j}^{n}(a_{j}^{k},m_{j}^{l})=v_{j}^{n}:j=1,2)=P(a_{j}^{k}G_{I}\oplus
m_{j}^{l}G_{O/I}\oplus B_{j}
^{n}=v_{j}^{n}:j=1,2 )$. For any choice of $g_{I}$ and $g_{O/I}$, there exists unique
$b_{j}^{n}:j=1,2$ such that $a_{j}^{k}g_{I}\oplus m_{j}^{l}g_{O/I}\oplus b_{j}
^{n}=v_{j}^{n}:j=1,2$. Since
$G_{I}$, $G_{O/I}$ and $B^{n}$ are mutually independent and uniformly distributed, the
probability in question is therefore
$\frac{q^{kn}q^{ln}}{q^{kn}q^{ln}q^{2n}}=\frac{1}{q^{2n}}$.\\
(c) Note that
\begin{eqnarray}
 \label{Eqn:IDontKnowWhatToNameThis}
 P\left(\substack{
V_{j}^{n}(0^{k},m_{j}^{l})=v_{j,0^{k}}^{n}:j=1,2,\\V^{n}(0^{k},{
\hatm^{l}}
)=v^{n}}\right)=P\left(\substack{
m_{j}^{l}G_{O/I}\oplus B_{j}^{n}=v_{j,0^{k}}^{n}:\\j=1,2,
\hatm^{l}G_{O/I}\oplus B^{n}=v^{n}}\right)\nonumber\\
=P\left(\substack{
m_{j}^{l}G_{O/I}\oplus B_{j}^{n}=v_{j,0^{k}}^{n}:j=1,2,\\
(\hatm^{l}\ominus (m_{1}^{l}\oplus m_{2}^{l}))G_{O/I}=v^{n}}\right)\nonumber
\end{eqnarray}
Since $\hatm^{l} \neq
m_{1}^{l}\oplus m_{2}^{l}$, there exists an index $t$ such that $\hatm_{t} \neq
m_{1t}\oplus m_{2t}$.
Therefore, given any set of rows
$\underline{g}_{O/I,1}\cdots,\underline{g}_{O/I,t-1},\underline{g}_{O/I,t+1},\cdots,
\underline{g}_{O/I ,l}$, there exists a unique selection for row $\underline{g}_{O/I,t}$
such that $(\hatm^{l}\ominus (m_{1}^{l}\oplus m_{2}^{l}))g_{O/I}=v^{n}$. Having chosen
this, choose
$b_{j}^{n}=v_{j,0^{k}}^{n}\ominus m_{j}^{l}g_{O/I}$. Since
$G_{I}$, $G_{O/I}$ and $B_{j}^{n}:j=1,2$ are mutually independent and uniformly
distributed, the
probability in question is therefore
$\frac{q^{(l-1)n}}{q^{ln}q^{2n}}=\frac{1}{q^{3n}}$.
\end{proof}

\begin{lemma}
\label{Lem:StatisticalIndependenceOfTransmittedCosetAndCompetingCodeword}
If generator matrices $G_{I}
\in \fieldq^{k \times n}$, $G_{O/I} \in \fieldq^{l \times n}$ and
$B_{j}^{n} \in \fieldq^{n}:j=1,2$ are
mutually independent and uniformly distributed over their respective range spaces, then
the pair of cosets $C_{j}(m^{l}_{j}):j=1,2$ is independent of
$V^{n}(\hata^{k},\hatm^{l})$ whenever
$\hatm^{l} \neq (m_{1}^{l}\oplus m_{2}^{l})$.
\end{lemma}
\begin{proof}
  Let $v_{j,a^{k}}^{n} \in \fieldq^{n}$ for each $a^{k} \in \fieldq^{k}$, $j=1,2$
and $\hatv^{n} \in \fieldq^{n}$. We need to prove
 \begin{eqnarray}
 \lefteqn{P( C_{j}^{n}(m_{j}^{l})=(v_{j,a^{k}}^{n}:a^{k} \in
\fieldq^{k}):j=1,2,}\nonumber\\&&~~~~~~~~~V^{n}(\hata^{k},{\hatm}^{l}
)=\hatv^ { n
} )\nonumber\\
&&= P( C_{j}^{n}(m_{j}^{l})=(v_{j,a^{k}}:a^{k} \in
\fieldq^{k}):j=1,2)\nonumber\\\label{Eqn:IdentityToBeEstablishedInIndependenceStep}
&&~~~~~~~~~P(
V^{n}(\hata^{k},{\hatm}^{l})=\hatv^{n} )
 \end{eqnarray}
for every choice of $v_{j,a^{k}} \in \fieldq^{n}:a^{k} \in
\fieldq^{k},j=1,2$ and $\hatv^{n}\in \fieldq^{n}$.

If (i) for some $j=1$ or $j=2$, $( v_{j,a^{k}\oplus \tilde{a}^{k}}^{n} - v_{j,{0^{k}}}^{n}
)
\neq (v_{j,a^{k}}^{n} - v_{j,0^{k}}^{n}) \oplus  (
v_{j,\tilde{a}^{k}}^{n} - v_{j,0^{k}}^{n})$ for any pair $a^{k}$,
$\tilde{a}^{k} \in \fieldq^{k}$, or (ii) $v_{1,a^{k}}^{n}-v_{1,0^{k}}^{n} \neq
v_{2,a^{k}}^{n}-v_{2,0^{k}}^{n}$ for some $a^{k} \in \fieldq^{k}$,
then LHS and first term of RHS are zero and equality holds.

Otherwise, LHS of (\ref{Eqn:IdentityToBeEstablishedInIndependenceStep}) is
\begin{eqnarray}
 \lefteqn{P\left(\substack{ C_{j}^{n}(m_{j}^{l})=(v_{j,a^{k}}^{n}:a^{k} \in
\fieldq^{k}):j=1,2,
V^{n}(\hata^{k},{\hatm}^{l})=\hatv^{n}
}\right)}\nonumber\\
&=& P\left(\substack{a^{k}G_{I}=v_{1,a^{k}}^{n}-v_{1,0^{k}}^{n}:a^{k} \in \fieldq^{k},
V_{j}^{n}(0^{k},m_{j}^{l})=v_{j,0^{k}}^{n}:j=1,2,\\V^{n}(0^{k},
 { \hatm^ { l}}
)=\hatv^{n}-(v_{1,\hata^{k}}^{n}-v_{1,0^{k}}^{n})}\right)\nonumber\\
\label{Eqn:LHSOfIdentityToBeEstablishedInIndependenceStep}
&=&P\left(\substack{a^{k}G_{I}=v_{1,a^{k}}^{n}-\\v_{1,0^{k}}^{n}:a^{k} \in
\fieldq^{k}}\right)P\left(\substack{
V_{j}^{n}(0^{k},m_{j}^{l})=v_{j,0^{k}}^{n}:j=1,2,\\V^{n}(0^{k},{
\hatm^{l}}
)=\hatv^{n}-(v_{1,\hata^{k}}^{n}-v_{1,0^{k}}^{n})}\right),
\end{eqnarray}
where we have used independence of $G_{I}$ and $(G_{O/I},B_{1}^{n},B_{2}^{n})$ in
arriving at (\ref{Eqn:LHSOfIdentityToBeEstablishedInIndependenceStep}). Similarly RHS
of (\ref{Eqn:IdentityToBeEstablishedInIndependenceStep}) is
\begin{eqnarray}
  \lefteqn{P\left(\substack{ C_{j}^{n}(m_{j}^{l})=(v_{j,a^{k}}^{n}:a^{k} \in
\fieldq^{k}):j=1,2}\right)P\left(\substack{
V^{n}(\hata^{k},{\hatm}^{l})=\hatv^{n}}
\right)}\nonumber\\
&=& P\left(\substack{a^{k}G_{I}=v_{1,a^{k}}^{n}-v_{1,0^{k}}^{n}:a^{k} \in
\fieldq^{k},\\
V_{j}^{n}(0^{k},m_{j}^{l})=v_{j,0^{k}}^{n}:j=1,2}\right)P\left(\substack{a^{k}G_{I}
\oplus \hatm^{l}G_{O/I}\oplus B^{n}=\\\hatv^ { n} }\right)\nonumber\\
\label{Eqn:SubstitutingForProbabilityOfACodeword}
&=&P\left(\substack{a^{k}G_{I}=v_{1,a^{k}}^{n}-\\v_{1,0^{k}}^{n}:a^{k} \in
\fieldq^{k}}\right)P\left(\substack{
V_{j}^{n}(0^{k},m_{j}^{l})=\\v_{j,0^{k}}^{n}:j=1,2}\right)\cdot
\frac{1}{q^{n}}\\
&=&P\left(\substack{a^{k}G_{I}=v_{1,a^{k}}^{n}-\\v_{1,0^{k}}^{n}:a^{k} \in
\fieldq^{k}}\right)P\left(\substack{
m_{j}^{l}G_{O/I}\oplus B_{j}^{n}=\\v_{j,0^{k}}^{n}:j=1,2}\right)\cdot
\frac{1}{q^{n}}\nonumber\\
\label{Eqn:RHSOfIdentityToBeEstablishedInIndependenceStep}
&=&P\left(a^{k}G_{I}=v_{1,a^{k}}^{n}-v_{1,0^{k}}^{n}:a^{k} \in
\fieldq^{k}\right)\cdot
\frac{1}{q^{3n}},
\end{eqnarray}
where (\ref{Eqn:SubstitutingForProbabilityOfACodeword}),
(\ref{Eqn:RHSOfIdentityToBeEstablishedInIndependenceStep}) follows from lemma
\ref{Lem:UniformDistributionAndPairwiseIndependenceOfCodewordsInARandomCoset}(a) and
(b) respectively. Comparing
simplified forms of
LHS in (\ref{Eqn:LHSOfIdentityToBeEstablishedInIndependenceStep}) and
RHS in (\ref{Eqn:RHSOfIdentityToBeEstablishedInIndependenceStep}), it suffices to prove
\begin{equation}
 \label{Eqn:SimplifiedFormOfTheIdentityToBeEstablishedInIndependenceStep}
 P\left(\substack{
V_{j}^{n}(0^{k},m_{j}^{l})=v_{j,0^{k}}^{n}:j=1,2,\\V^{n}(0^{k},{
\hatm^{l}}
)=\hatv^{n}-(v_{1,\hata^{k}}^{n}-v_{1,0^{k}}^{n})}\right)=\frac{1}{q^{3n}}.\nonumber
\end{equation}
This follows from lemma
\ref{Lem:UniformDistributionAndPairwiseIndependenceOfCodewordsInARandomCoset}(c)
\end{proof}
We emphasize consequence of lemma
\ref{Lem:StatisticalIndependenceOfTransmittedCosetAndCompetingCodeword} in the following.
\begin{remark}
 \label{Rem:KeyIndependenceStep}
If $\hatm^{l} \neq hs_{1}^{n} \oplus hs_{2}^{n}$, then conditioned on the event $\left\{ 
S_{j}^{n}=s_{j}^{n}:j=1,2\right\}$, received vector $Y^{n}$ is statistically independent
of $V^{n}(\hata^{k},{\hatm^{l}})$ for any $\hata^{k} \in \SourceAlphabet^{k}$. We
establish truth of this statement in the sequel. Let $\mathcal{C}$ denote the set of all
ordered $\left|\SourceAlphabet\right|^{k}$-tuples of vectors in $\SourceAlphabet^{n}$.
Observe that,
\begin{eqnarray}
 \lefteqn{P\left(
\substack{\underlines^{n}=\underlines^{n},Y^{n}=y^{n},\\V^{n}(\hata^{k},\hatm^{l})=\hatv^{
n } }
\right)=\sum_{C_{1} \in \mathcal{C} }
\sum_{C_{2} \in \mathcal{C}}
P\left(\substack{
\underlines^{n}=\underlines^{n},C_{j}(hs_{j}^{n})=C_{j}:j=1,2,\\V^{n}(\hata^{k},\hatm^
{l}
)=\hatv^{n},Y^{n}=y^{n}}\right) }\nonumber\\
&=& \sum_{C_{1} \in \mathcal{C}_{1} }
\sum_{C_{2} \in \mathcal{C}_{2}} P\left(  \substack{\underlines^{n}=\underlines^{n}}
\right)P\left( \substack{ C_{1}(hs_{1}^{n})=C_{1}\\C_{2}(hs_{2}^{n})=C_{2}}
\right)P\left(\substack{V^{n}(\hata^{k},\hatm^{l}
)=\hatv^{n}} \right)\nonumber\\
\label{Eqn:UsingStatisticalIndependenceOfCosetsCorrToMessageAndCompetingCodeword}
&&~~~~~~~~~~~~~~~\cdot P\left(
 Y^{n}=y^{n} |
\substack{C_{j}(hs_{j}^{n})=C_{j}:j=1,2\\\underlines^{n}=\underlines^{n}
} \right)\\
&=&\sum_{C_{1} \in \mathcal{C}_{1} } \sum_{C_{2} \in \mathcal{C}_{2}}
P\left(\substack{
\underlines^{n}=\underlines^{n},Y^{n}=y^{n},\\C_{j}(hs_{j}^{n})=C_{j}:j=1,2}\right)
P\left(\substack{V^{n}(\hata^{k},\hatm^{l}
)=\hatv^{n}} \right)
\nonumber\\
&=&P\left(\underlines^{n}=\underlines^{n},Y^{n}=y^{n}\right)
P\left(V^{n}(\hata^{k},\hatm^{l}
)=\hatv^{n}\right)\nonumber
\end{eqnarray}
We have used (a) independence of $\underlines^{n}$ and random objects that characterize
the
codebook, (b) independence of $V^{n}(\hata^{k},\hatm^{l})$ and
$(C_{j}(hs_{j}^{n}):j=1,2)$ (lemma
\ref{Lem:StatisticalIndependenceOfTransmittedCosetAndCompetingCodeword}),
(c) $(\mu_{1}(hs_{1}^{n}),\mu_{2}(hs_{2}^{n}))$ being a function of
$(C_{1}(hs_{1}^{n}),C_{2}(hs_{2}^{n}))$, is conditionally
independent of $V^{n}(\hata^{k},\hatm^{l})$ given
$(C_{1}(hs_{1}^{n}),C_{2}(hs_{2}^{n}))$ in arriving at
(\ref{Eqn:UsingStatisticalIndependenceOfCosetsCorrToMessageAndCompetingCodeword}).
Moreover, since $P(V^{n}(\hata^{k},\hatm^{l})=\hatv^{n})=\frac{1}{\left|\SourceAlphabet
\right|^{n}}$, we have
$P\left(\underlines^{n}=\underlines^{n},Y^{n}=y^{n},\\V^{n}(\hata^{k},\hatm^{l})=\hatv^{
n } 
\right)=\frac{1}{\left|\SourceAlphabet\right|^{n}}P(\underlines^{n}=\underlines^{n},Y^{n}
=y^{n} )$.
\end{remark}
We are now equipped to derive an upper bound on $P(\epsilon_{3})$. Observe that
\begin{eqnarray}
\label{Eqn:UpperboundingDecodingErrorProbability}
\lefteqn{ \textstyle P(\epsilon_{3}) \leq P\left(\underset{\hata^{k} \in
\SourceAlphabet^{k}}{\bigcup}
\underset{\underlines^{n}=\underlines^{n}}{\bigcup}\underset{\substack{\hatm^{l} \neq\\
h(s_{1}^{n}\oplus s_{2}^{n} ) }}{\bigcup}\left\{ \substack{ 
(V^{n}(\hata^{k},\hatm^{l}),Y^{n}) \in \\
T_{\eta_{1}}(p_{V_{1} \oplus
V_{2}, Y}),\underlines^{n}=\underlines^{n} }  \right\}\right)}\nonumber\\
 &\leq&  \underset{ \substack{\hata^{k} \in
\SourceAlphabet^{k},\\\underlines^{n}=\underlines^{n} }}{\sum}
\underset{\substack{\hatm^{l} \neq\\
h(s_{1}^{n}\oplus s_{2}^{n} ) }}{\sum}\underset{\substack{   y^{n}\in
T_{\eta_{1}} (Y) ,v^{n} \in\\ \small
T_{\eta_{1}}(V_{1}\oplus V_{2}|y^{n} )  } }{\sum} \!\!\!\!\!\!\!\!
P\left(\substack{V^{n}(a^{k},{\hatm^{l}})=v^{n}\\\underlines^{n}=\underlines^{n},Y^{n}=y^{
n}}
\right)\nonumber\\
\label{Eqn:IndependenceFollowingFromAboveRemark}
\textstyle &\leq& \underset{ \substack{\hata^{k} \in
\SourceAlphabet^{k},\\\underlines^{n}=\underlines^{n} }}{\sum}
\underset{\substack{\hatm^{l} \neq\\
h(s_{1}^{n}\oplus s_{2}^{n} ) }}{\sum}\underset{\substack{   y^{n}\in
T_{\eta_{1}} (Y) ,v^{n} \in\\ \small
T_{\eta_{1}}(V_{1}\oplus V_{2}|y^{n} )  } }{\sum}  \!\!\!\!\!\!\!\!
P\left(\substack{V^{n}(a^{k},\hatm^{l})\\=v^{n}}\right)P\left(\substack{\underlines^{n}
=\underlines^ { n } , \\Y^ { n } =y^ { n }}
\right)\nonumber\\
\label{Eqn:IndependenceFollowingFromAboveRemark2}
\textstyle &\leq&  \underset{\hata^{k} \in
\SourceAlphabet^{k}}{\sum}
\underset{\substack{\hatm^{l} \neq\\
h(s_{1}^{n}\oplus s_{2}^{n} ) }}{\sum}\underset{\substack{   y^{n}\\\in
T_{\eta_{1}} (Y)  }}{\sum}  \underset{\substack{v^{n} \in\\ \small
T_{\eta_{1}}(V_{1}\oplus V_{2}|y^{n} ) }}{\sum}
\frac{P(Y^{n}=y^{n})}{\left| \SourceAlphabet\right|^{n}}\nonumber\\
\textstyle &\leq&  \underset{\substack{y^{n}\\\in
T_{\eta_{1}} (Y)}   }{\sum}\!\!\!  \frac{\left|
\SourceAlphabet\right|^{k+l}|T_{\eta_{1}}(V_{1} \oplus
V_{2}|y^{n} )|}{\left| \SourceAlphabet\right|^{n}} \nonumber\\
\label{Eqn:UpperboundOnDecodingErrorProbability}&\leq&
\textstyle\exp \left\{-n \log \left| \SourceAlphabet\right| \left(1-\frac{H(V_{1}\oplus
V_{2}|Y)+3\eta_{1}+k+l}{\log \left| \SourceAlphabet\right|}\right)\right\}.\end{eqnarray}
where (\ref{Eqn:UpperboundOnDecodingErrorProbability}) follows
from the uniform bound of $\exp\left\{n\left(H(V_{1} \oplus
V_{2}|Y)+3\eta_{1}\right)\right\}$ on $|T_{\eta_{1}}(V_{1} \oplus V_{2}|y^{n}
)|$ for any $y^{n}\in
T_{\eta_{1}}(Y)$, $n \geq N_{6}(\eta)$ (Conditional frequency typicality) for $n\geq
N_{6}(\eta)$.
\section{Concluding Remarks}
Having decoded the sum of sources, we ask whether it would be possible to decode an arbitrary function of the sources using the above techniques? The answer is yes and the technique involves `embedding'. Example \ref{Ex:TernarySourcesAndDecoderIntersetedInParityOfThePair} illustrates embedding and a framework is proposed in a subsequent version of this article. This leads us to the following fundamental question. The central element of the technique presented above was to decode the \textit{sum} of transmitted codewords and use that to decode sum of KM message indices. If the MAC is `far from additive', is it possible to decode a different bivariate function of transmitted codewords and use that to decode the desired function of the sources? The answer to the first question is yes. Indeed, the elegance of joint typical encoding and decoding enables us reconstruct other `well behaved' functions of transmitted codewords. We recognize that if codebooks take values over a finite field and were closed under addition, it was natural and more efficient to decode the sum. On the other hand, if the codebooks were taking values over an algebraic object, for example a group, and were closed with respect to group multiplication, it would be natural and efficient to decode the product of transmitted codewords. Since, we did not require the MAC to be linear in order to compute the sum of transmitted codewords, we will not require it to multiply in order for us to decode the product of transmitted codewords. We elaborate on this in a subsequent version of this article.
\section*{Acknowledgment}
This work was supported by NSF grant CCF-1111061.
\bibliographystyle{../sty/IEEEtran}
{
\bibliography{ComputationOverMAC.bbl}
}
\end{document}